\documentclass[12pt,oneside]{article}
\usepackage[immediate]{silence}
\usepackage{sectsty}
\usepackage[letterpaper,top=1in, bottom=1in, left=1in, right=1in]{geometry}
\usepackage{setspace}

\usepackage{microtype}

\usepackage{ulem}
\usepackage{bm}
\usepackage{amsfonts}
\usepackage{amsmath}
\usepackage{mathtools}
\usepackage{amsthm}
\usepackage{thmtools}
\usepackage{textcomp}
\usepackage{amssymb}
\usepackage[svgnames]{xcolor}
\usepackage[sc,labelsep=period]{caption}
\allowdisplaybreaks

\usepackage{graphicx}
\usepackage{csquotes}
\usepackage[authoryear,round,longnamesfirst,numbers]{natbib}
\usepackage{hyperref}
\hypersetup{
pdftitle={Self-Enforced Job Matching},
pdfauthor={Ce Liu, Ziwei Wang, and Hanzhe Zhang},
pdfkeywords={dynamic stability, repeated game, two-sided matching}
}
\hypersetup{
colorlinks,breaklinks,naturalnames,
citecolor=DarkBlue,urlcolor=DarkBlue,linkcolor=DarkBlue
}
\usepackage[capitalize,nameinlink]{cleveref}
\usepackage{titlesec}

\usepackage[inline,shortlabels]{enumitem}
\usepackage{xargs}
\usepackage{comment}
\usepackage{tikz}
\usetikzlibrary{arrows.meta}
\usepackage{booktabs}
\usepackage{array} 
\usepackage{etoolbox}

\bibpunct[, ]{(}{)}{;}{a}{}{,}
\MakeOuterQuote{"}
\theoremstyle{plain}

\newtheorem{assumption}{Assumption}

\newtheorem{lemma}{Lemma}

\newtheorem{proposition}{Proposition}

\theoremstyle{definition}
\newtheorem{definition}{Definition}
\theoremstyle{remark}

\sectionfont{\color{DarkRed}}
\subsectionfont{\color{DarkRed}}
\subsubsectionfont{\color{DarkRed}}
\renewcommand{\paragraph}[1]{{\flushleft \textbf{\color{DarkRed}#1 }}}

\renewcommand\thmcontinues[1]{Continued}

\AtEndEnvironment{proof}{\setcounter{claim}{0}}

\newcommand{\Enviro}{\Phi}
\newcommand{\firms}{\mathcal{F}}
\newcommand{\firm}{f}
\newcommand{\workers}{\mathcal{W}}
\newcommand{\worker}{w}
\newcommand{\types}{\Theta}
\newcommand{\type}{\theta}
\newcommand{\wage}{p}
\newcommand{\smatchings}{M}
\newcommand{\history}{h}
\newcommand{\histories}{\mathcal{H}}

\renewcommand{\Re}{\mathbb{R}}

\newcommand{\abs}[1]{ \left| #1 \right| }

\newcommand{\conv}[1]{\text{co}\left(#1\right)} 
\newcommand{\term}[1]{\textit{#1}}

\newcommand{\fhistories}{\overline{\histories}}

\renewcommand{\emph}[1]{\textit{#1}}
\renewcommand{\exp}{\mathbb{E}}
\renewcommand{\tilde}{\widetilde}
\renewcommand{\hat}{\widehat}

\renewcommand{\omega}{s}
\renewcommand{\Omega}{S}

\DeclareMathOperator*{\argmin}{arg\,min}
\DeclareMathOperator*{\argmax}{arg\,max}
\newcommand*{\supp}{\mathrm{supp}}

\newenvironment{proofof}[1]{\par\bigskip
   \noindent \textbf{Proof of #1.} } {\qed \par \bigskip}

\newcommand{\hchi}{\raise0.4ex\hbox{$\chi$}}

\DeclareFontFamily{U}{mathx}{\hyphenchar\font45}
\DeclareFontShape{U}{mathx}{m}{n}{
      <5> <6> <7> <8> <9> <10>
      <10.95> <12> <14.4> <17.28> <20.74> <24.88>
      mathx10
      }{}
\DeclareSymbolFont{mathx}{U}{mathx}{m}{n}
\DeclareMathSymbol{\bigtimes}{1}{mathx}{"91}

\usepackage[bottom,multiple,splitrule]{footmisc}
\title{Self-Enforced Job Matching\thanks{We thank Nageeb Ali, Mariagiovanna Baccara, Vincent Crawford, Henrique de Oliviera, Jon Eguia, Ed Green, Scott Kominers, Maciej Kotowski, SangMok Lee, Xiao Lin, Elliot Lipnowski, Qingmin Liu, Arijit Mukherjee, Marzena Rostek, Joel Sobel, Quan Wen, Kun Zhang, and audiences at various seminars for their helpful comments.}}

\begin{document}

\author{Ce Liu\thanks{Department of Economics, Michigan State University. E-mail: \href{mailto: celiu@msu.edu}{celiu@msu.edu}.} \and Ziwei Wang\thanks{School of Economics and Management, Wuhan University. E-mail: \href{mailto: zwang.econ@gmail.com}{zwang.econ@gmail.com}.} \and Hanzhe Zhang\thanks{Department of Economics, Michigan State University. E-mail: \href{mailto: hanzhe@msu.edu}{hanzhe@msu.edu}.}}

\date{\today}

\maketitle
\thispagestyle{empty}

\begin{abstract}

\noindent 
The classic two-sided many-to-one job matching model assumes that firms treat workers as substitutes and workers ignore colleagues when choosing where to work. Relaxing these assumptions may lead to nonexistence of stable matchings. However, matching is often not a static allocation, but an ongoing process with long-lived firms and short-lived workers. We show that stability is always guaranteed dynamically when firms are patient, even with complementarities in firm technologies and peer effects in worker preferences. While no-poaching agreements are anti-competitive, they can maintain dynamic stability in markets that are otherwise unstable, which may contribute to their prevalence in labor markets.

\end{abstract}

\vfill

\clearpage

\setcounter{page}{1}

\setstretch{1.25}

\clearpage

\section{Introduction} \label{section: intro}

Consider two-sided many-to-one matching markets in which each firm can employ workers at individually negotiated wages. In their seminal work, \cite{kelsocrawford1982} demonstrated that stable matchings exist in such markets when firms' production technologies satisfy gross substitutability and workers' preferences do not exhibit peer effects. Gross substitutability holds when no workers serve as complementary inputs in a firm's production, while the no-peer-effect assumption requires that workers are indifferent toward their colleagues' identities and characteristics. This framework has been employed across a wide range of applications, including auctions, labor markets, and housing.

However, many markets feature complementarities and peer effects. In industries such as manufacturing, healthcare, and software development, it is common for tasks to require the collaboration of workers with complementary skills. Also, in various labor markets---e.g., academia---the set of colleagues is an important consideration. Although recent literature has obtained positive existence results by modeling large markets or considering alternative assumptions on preferences, accommodating arbitrary production technologies and peer effects has been challenging.

In this paper, we propose an approach that does not rely on restrictions on market size, technologies, or preferences. We observe that matching is often an ongoing process in which long-lived players on one side of the market interact with short-lived players or objects on the other side over time. This dynamic is salient in many matching situations, including entry-level hiring, buyer-seller relationships, and multi-unit auctions. Given this observation, we consider the stability of matching as an unfolding process rather than a static allocation. We find that wage-suppressing no-poaching agreements always help maintain a dynamically stable matching process when firms are patient, even in the presence of complementarity and peer effects that destabilize static matchings. The key idea is that dynamic incentives can act as both carrot and stick to deter firms' deviations. We provide an illustrative example at the end of this section.

No-poaching agreements appear in various matching markets \citep{krueger2022theory}; e.g., those of IT workers, fast food servers, animation studio artists, and college admissions. These agreements have not only generated considerable controversy, but also come under intense scrutiny in antitrust regulations.\footnote{For their no-poaching practices, tech giants settled for \$415 million (\textit{US v.~Adobe Systems Inc., et al.}), and Duke and UNC paid \$73.5 million (\textit{Seaman v.~Duke University} and \textit{Binotti v.~Duke University}). For the fast food industry, see \textit{Deslandes v.~McDonald’s USA LLC} and \textit{Conrad v.~Jimmy John’s Franchise LLC}.} An ongoing debate concerns universities' use of a ``consensus approach'' to allocate financial aid to students. Despite antitrust exemptions granted to universities under the Improving America's Schools Act of 1994, this practice is currently the subject of civil antitrust litigation (\textit{Henry, et al.~v.~Brown University, et al.}). Our finding adds a new perspective to this debate: In a matching market with significant complementarities and peer effects, no-poaching agreements can help maintain stability in the absence of static stable matchings. Because stability is crucial for the effective functioning of markets \citep{roth2018marketplaces}, prohibiting such agreements could lead to market failure.

From a technical standpoint, our setting differs from standard repeated games due to the cooperative nature of the stage game. In standard repeated games, the noncooperative stage game is guaranteed to have a Nash equilibrium; playing a Nash equilibrium at every history delivers a subgame perfect Nash equilibrium. In our setting, the stage game may not have a stable matching. Therefore, unlike in standard repeated games, the existence of a dynamically stable matching process cannot be guaranteed by simply playing a stable stage matching at every history. Instead, we leverage random serial dictatorship to explicitly construct a dynamically stable matching process.

\paragraph{Related Literature.} 
The literature on static job matching has guaranteed the existence of stable matching by (i) imposing restrictions on preferences \citep{hatfieldmilgrom, sunyang2006, hatfieldkojima2008, pycia2012stability, HatfieldEtAl2013JPE, RostekYoder2020Ecma, KojimaSunYu2020AER, KojimaSunYu2023REStud, PyciaYenmez2023REStud}; (ii) considering large markets \citep{KojimaPathakRoth2013QJE, AshlagiEtAl2014OR, azevedohatfield2018, CheKimKojima2019Ecma}; or (iii) making minimal adjustments to quotas \citep{NguyenVohra2018AER}. In this paper, we propose a different approach based on firms' dynamic incentives.

Our paper also contributes to the literature on dynamic stability in matching: See, for example, \cite*{damianolam2005},
\citet*{du2016rigidity}, \cite*{kadamkotowski20182, kadamkotwoski2018},
\cite*{altinok2020}, \cite*{kotowski2020}, \cite*{kurino2020}, and \citet*{doval2022}. The work most closely related to this paper is \cite{liu2022}, who studies repeated matching markets \textit{without} 
transfers, while imposing assumptions so that static stable matchings exist in the stage game. This paper studies matching markets \textit{with} transfers using a similar solution concept; however, in contrast to \cite{liu2022}, we show that dynamic stability can be restored even when static stability fails. More broadly, \citet{aliliu} explore when and how dynamic incentives deter coalitional deviations in a general repeated cooperative games framework, but focus on settings where all players are long-lived.

\paragraph{An Example.}
Two long-lived firms $\firm_1$ and $\firm_2$ each offer two internship positions every period. Each firm treats workers as complements: It generates a revenue of \$$6$ only when both of its vacancies are filled and is unproductive otherwise. On the other side of the market, every period, three new identical workers---$w_1$, $w_2$, and $w_3$---look for positions. For simplicity, assume that workers' payoffs are equal to the wages they receive. If the matching market in every period is treated as an isolated one-shot interaction, there is no \textit{static} stable matching: In any static matching, at most one firm can be productive, and for the firms to break even, at most one worker can earn more than \$$3$. An unproductive firm can then form a blocking coalition with two workers, each of whom currently earns no more than \$$3$, and split the gain.

Instead of static stable matchings, consider the following history-dependent matching process (illustrated in \Cref{fig:example}). In each period, the market is in one of four possible states: two wage-suppressing collusion states C1 and C2 and two punishment states P1 and P2. Each collusion state has a dominant firm, $f_1$ in C1 and $f_2$ in C2; firms $f_1$ and $f_2$ are punished in states P1 and P2, respectively. The market starts at and remains in state C1 until a firm deviates. Whenever a firm deviates, the market transitions to the punishment state for that firm, where it remains for four periods. After this punishment phase, the market enters the collusion state in which the non-deviating firm assumes the role of the dominant firm, and stays there indefinitely until a firm deviates.

\begin{figure}[ht]
   \captionsetup{width=.75\linewidth}
{    \centering
\begin{tikzpicture}
    \node[shape=circle,draw=black, minimum size=1.6cm,line width=0.5mm, align=center, font=\small] (C1) at (0,2) {C1\\$(4,2)$};
    \node[shape=circle,draw=black, minimum size=1.6cm,line width=0.5mm, align=center, font=\small] (C2) at (5,2) {C2\\$(2,4)$};
    \node[shape=rectangle,draw=black, minimum size=1.4cm,line width=0.5mm, align=center, font=\small] (P1) at (5,-0.5) {P1\\$(0,-6)$};
    \node[shape=rectangle,draw=black, minimum size=1.4cm,line width=0.5mm, align=center, font=\small] (P2) at (0,-0.5) {P2\\$(-6,0)$};
    
    \node[right of=P1, xshift=0.1cm] {$\times 4$};
    \node[left of=P2, xshift=-0.1cm] {$4 \times$};
    
    \path[->,  loop above, out=130, in=50, looseness=3.5, line width=0.5mm, -stealth] (C1) edge (C1);
    \path[->,  loop above, out=130, in=50, looseness=3.5, line width=0.5mm, -stealth] (C2) edge (C2);
    
    \path[->,  line width=0.5mm, out=90, in=-90, looseness=0, -stealth] ([xshift=0.1cm]P2.north) edge ([xshift=0.1cm]C1.south);
    \path[->, line width=0.5mm, out=90, in=-90, looseness=0, -stealth] ([xshift=0.1cm]P1.north) edge ([xshift=0.1cm]C2.south);

    \path[->,red, line width=0.5mm,  -stealth] (C1) edge (P1);
    \path[->,red, line width=0.5mm,  -stealth] ([yshift=-0.1cm]P2.east) edge ([yshift=-0.1cm]P1.west);
    \path[->,red, line width=0.5mm,  -stealth] ([xshift=-0.1cm]C2.south) edge ([xshift=-0.1cm]P1.north);    
    
    \path[->,blue, line width=0.5mm,  -stealth] (C2) edge (P2);
    \path[->,blue, line width=0.5mm,  -stealth] ([yshift=0.1cm]P1.west) edge ([yshift=0.1cm]P2.east);
    \path[->,blue, line width=0.5mm,  -stealth] ([xshift=-0.1cm]C1.south) edge ([xshift=-0.1cm]P2.north);    
\end{tikzpicture}
    \caption{A dynamically stable matching process. The first number in parentheses is the per-period payoff of $f_1$, while the second number is the payoff of $f_2$. Black arrows represent transitions when no firm deviates, red arrows represent transitions after $\firm_1$'s deviations, and blue arrows represent transitions after $\firm_2$'s deviations.\label{fig:example}}
}
\end{figure}
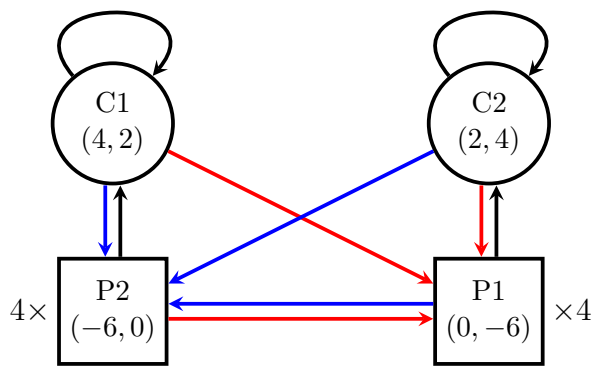

At the beginning of every period in collusion states C1 and C2, a biased coin toss determines which firm secures the hiring rights for the current period. The dominant firm wins the coin toss with a probability of $2/3$, while the nondominant firm has a winning probability of $1/3$. The firm that wins the coin toss hires two workers at zero wage, while the losing firm temporarily shuts down. Therefore, the dominant firm obtains \$$4$ on average, while the other obtains \$$2$. Notice that the collusion states feature a no-poaching agreement---in each period, the non-hiring firm refrains from soliciting workers from the other firm, thus forgoing an immediate gain in the current period. 
In punishment states P1 and P2, the firm being punished shuts down, while the punishing firm hires two workers each at a wage of \$$6$. \Cref{fig:example} depicts firms' payoffs in each state.

To see that this matching process is dynamically stable, we use a one-shot deviation principle (\cref{lemma: one-shot-deviation}), which establishes dynamic stability by checking two requirements at every history: (1) no worker wishes to unilaterally leave her matched firm to be unemployed and (2) no firm has a profitable one-shot deviation with a group of workers who also find this deviation profitable. 
Note that the first requirement is satisfied in every state of the matching process, since all workers are weakly better off than being unemployed. We next verify that the second requirement is also met in every state.

In punishment states, the punished firm cannot find a profitable deviation in the stage game. The punishing firm can gain a profit in the current period by deviating, but by doing so it would forgo the advantage of being the dominant firm in future collusive periods. As a result, for high values of $\delta$, the punishing firm does not find one-shot deviations profitable.

In collusion states, regardless of which firm has the hiring rights for the period, the dominant firm finds no profitable one-shot deviations when $\delta$ is high: Any deviation would lead to the loss of its position as the dominant firm, and thus a lower long-run payoff. The nondominant firm also has no incentive to deviate: Even though it could potentially obtain a current-period gain of \$$6$ by poaching workers from the winning firm, this would be followed by a net loss of \$$2$ for each of the next four periods (it gets \$$0$ from shutting down, compared with \$$2$ in the collusive state). When firms are patient, the loss outweighs the potential gain, so the nondominant firm also has no profitable one-shot deviations.

\section{Model} \label{section: general model}

\paragraph{Players.} At the beginning of every period $t = 0, 1, 2, \ldots$, a new generation of {workers} enters the market to match with a fixed set of long-lived firms. Let $\firms$ denote the set of firms, where $|\firms|\geq 2$. Matching is many-to-one: In every period, each firm $\firm \in \firms$ has $q_{\firm} > 0$ hiring slots to fill. Workers are short-lived and remain in the market for only one period. For expositional convenience, we use the same notation $\workers$ to denote the set of workers in every period; however, it is important to note that each worker $\worker \in \workers$ can have a different \term{type} in every period.\footnote{Our model can accommodate fluctuations in the number of workers across periods by introducing undesirable types, so that matching with an undesirable worker renders all of her partners strictly worse off (before transfer) than staying unmatched.} A finite set $\types_{\worker}$ contains all possible types $\type_{\worker}$ of worker $\worker$, and $\types=\bigtimes_{\worker \in \workers}\types_{\worker}$ is the set of type profiles of all workers. Workers' type profile $\type \in \types$ is randomly drawn from a distribution $\pi\in\Delta(\types)$ for each generation.

Each firm $\firm$ has a stage-game revenue function $\tilde{u}_{\firm}: 2^{\workers}\times \types \rightarrow \mathbb{R}$ defined over subsets of workers and type profiles.\footnote{While it may be reasonable to assume  that a firm's revenue depends solely on the types of workers it employs, none of our results require this assumption. Hence, for notational simplicity, we allow revenue to depend on the entire profile of worker types. The same applies to workers' payoff functions.}  
We normalize the revenue of staying unmatched, $\tilde{u}_{\firm}(\emptyset,\type)$, to $0$ for every $\type\in \types$.
Note that for each type profile, we do not require firms' revenue functions to satisfy the gross substitutes condition. Firms share a common discount factor $\delta$ and evaluate a sequence of flow utilities through exponential discounting. 

Each worker cares about both her employer and her colleagues, which we will refer to collectively as her \emph{work environment}. Let $\Enviro_\worker=( \firms \times   2^{\workers\backslash\{\worker\}}) \cup \{ (\emptyset, \emptyset )\}$ denote the set of possible work environments of worker $\worker$, where $(\emptyset, \emptyset)$ represents staying unmatched. Each worker $\worker$ has utility function $\tilde{v}_{\worker}: \Phi_w \times \types  \rightarrow \mathbb{R}$ over work environments.  For every worker $\worker \in \workers$ and type profile $\type \in \types$, $\tilde{v}_{\worker}(\emptyset, \emptyset, \type)$ is normalized to $0$. 

\paragraph{Stage Matching.} The outcome in every period is a static many-to-one matching among the firms and workers. Formally, a stage matching $m= (\phi, {\wage})$ is described by an \term{assignment} $\phi$ and a \term{wage vector} ${\wage}$. In particular, $\phi$ is a mapping defined on the set $\firms \cup \workers$ such that (i) for every $\worker \in \workers$, $\phi(\worker) \in \Enviro_w$; (ii) for every $\firm\in \firms$, $\phi(\firm) \subseteq {\workers}$ and $\abs{\phi(\firm)} \le q_{\firm} $; and (iii) for every $\worker\in\workers$ and every $\firm\in\firms$, $\worker \in \phi(\firm)$ if and only if $\phi(\worker) = (\firm, W')$ for some $W' \in 2^{\workers\backslash\{\worker\}}$. The {wage vector} ${\wage}=(\wage_{\firm \worker}) \in \Re^{\abs{\firms } \times {\abs{\workers}}}$ describes the transfer from firms to workers. 
We assume that any nonzero transfer occurs only between a firm and its own employees: $\wage_{\firm \worker} =0$ for every $\worker \notin \phi(\firm)$.

Players have quasilinear utilities: For every stage matching $m = (\phi, {\wage})$ and type profile $\type$,
\[
	u_{\firm}(m,\type) \equiv  \tilde{u}_{\firm}(\phi( \firm ),\type) - \sum_{\worker'\in \workers } \wage_{\firm \worker'} \;\text{ and }\;
	v_{\worker}(m,\type) \equiv \tilde{v}_{\worker}(\phi(\worker), \type) +   \sum_{\firm' \in \firms } \wage_{\firm' \worker} 
\]
are the stage-game payoffs received by firm $\firm$ and worker $\worker$, respectively.\footnote{Our results do not hinge on quasilinearity and hold as long as utilities are monotone and unbounded in wages.}

A stable stage matching is immune to three kinds of deviation: (i) a unilateral deviation by a firm $\firm$ who fires a subset of its employees and leaves those positions unfilled; (ii) a unilateral deviation by a worker $\worker$ who leaves her employer and remains unmatched; and (iii) a coalitional deviation $(\firm,W', {\wage}'_{\firm}) \in \firms \times 2^{\workers} \times \Re^{\abs{\workers}}$ with $|W'| \le q_{\firm}$ and $\wage'_{\firm \worker}=0$ for every $\worker\notin W'$, where firm $\firm$ and workers $W'$ match at wages specified in ${\wage}'_{\firm}$, abandoning any other pre-existing match partners. A deviation is profitable if all participants strictly prefer the deviation to the original matching $m$. A stage matching is \term{individually rational} if no unilateral deviations are profitable, and is \term{stable} if none of the three kinds of deviations above is profitable.

When studying the stability of static matchings, there is no need to specify how other players will be matched after a deviation by a coalition. However, in our dynamic setting, players' future behavior is influenced by past histories, so to study the stability of matching processes, we need to specify the realized stage-game outcome after a deviation. To this end, we adopt the following assumption.

\begin{assumption} \label{assumption: no further deviation}
After coalitional deviation $(\firm, W', {\wage}'_{\firm})$ from $m=(\phi,\wage)$, let $[m,(\firm, W', {\wage}'_{\firm})] \in M$ denote the resulting stage matching and $\phi'$ denote the assignment in $m' = [m,(\firm, W', {\wage}'_{\firm})]$. Assume $\phi'(\firm) = W'$ and $\phi' (\firm') = \phi (\firm')\backslash W' $ for every $\firm' \ne \firm$.
\end{assumption}

\cref{assumption: no further deviation} says that in the stage matching that results from a coalitional deviation, the deviators are matched together, players abandoned by the deviators remain unmatched, and those untouched by the deviation remain matched with their original partners. The sole purpose of this assumption is to ensure ``perfect monitoring'': When a matching $m$ is blocked by a coalition $(\firm, W', {\wage}'_{\firm})$, the firm in the deviating coalition is identifiable. \cref{lemma: identifiability of manipulator} formally proves this identifiability property from \cref{assumption: no further deviation}. Any alternative assumption that delivers the same identifiability property will not change our results.

\paragraph{Repeated Matching.} 
To convexify players' stage-game payoffs, we employ a public randomization device on the unit interval $\Omega=[0,1]$ endowed with the Lebesgue measure. While the use of public randomization streamlines our proofs, our results do not depend on it.\footnote{For example, we could follow arguments in \cite{sorin1986} and \cite{fudenbergmaskin1991} to convexify players' payoffs using sequences of play instead.}

The timing in each period is as follows. First, a new cohort of workers arrive. A type profile $\type \in \types$ is drawn from the distribution $\pi$, and a public randomization $\omega \in \Omega$ is realized. Based on the realized $(\type, \omega)$, a stage matching is recommended for active players in the market. Players then decide whether to deviate from this recommendation, which determines the outcome of the stage game.

A $t$-period ex ante history $\overline{\history} = (\type_{\tau}, \omega_{\tau}, m_{\tau})_{\tau = 0}^{t-1}$ specifies a sequence of past realizations of the type profile, the public correlating device, and the stage matching up to period $t-1$. We write $\fhistories_t$ for the set of all $t$-period ex ante histories, with $\fhistories_0 = \{\emptyset \}$ being the singleton set comprising the null history. Let $\fhistories \equiv \bigcup_{t=0}^\infty \fhistories_t$ be the set of all ex ante histories. Moreover, let $\histories_t \equiv \fhistories_t \times \types \times \Omega$ denote the set of $t$-period ex post histories, and $\histories \equiv \fhistories \times \types \times \Omega$ the set of all ex post histories. 

A \term{matching process} $\mu: \histories \rightarrow M$ specifies a stage matching for every ex post history. It represents a shared understanding among players regarding how past histories impact future employment. For every ex post history $h\in\histories$, we let $\mu(\firm|h)$ and $\mu(\worker|h)$ denote the matching partners for firm $\firm$ and worker $\worker$ in the stage matching $\mu(\history)$, respectively.

Let $\fhistories_\infty = (\types \times \Omega\times M)^\infty$ be the set of outcomes $\history_\infty$ of the repeated matching game. Let $m_{t}(\history_\infty)$ denote the stage matching in the $t$-th period of $\history_\infty$. Following every $t$-period (ex ante or ex post) history $\hat{\history} \in \fhistories \cup \histories$, let
\begin{equation*}
U_{\firm}(\hat{\history}\,|\,\mu) \equiv (1-\delta) \exp_{\mu}\Big[  \sum_{\tau = t}^\infty \delta^{\tau - t} u_{\firm} ( m_\tau(\history_\infty) , \type_\tau ) \, \Big| \, \hat{\history} \, \Big]
\end{equation*}
denote the continuation payoff firm $\firm$ obtains from $\mu$ following $\hat{\history}$, where the expectation is taken with respect to the measure over $\fhistories_\infty$ induced by $\mu$ conditional on $\hat{\history}$.

\paragraph{{Deviation Plan}.}  We make two observations that allow us to tractably analyze firms' deviations in repeated matching. Recall that there are three kinds of deviations within a period. We first observe that a unilateral deviation by firm $\firm$ to fire a subset of its employees is equivalent to a deviation by the coalition consisting of $\firm$ and the workers that remain with $\firm$. It is therefore without loss to focus on unilateral deviations by workers and deviations by coalitions consisting of a firm and a set of workers in the stage game.

Our second observation is that as a long-lived player, a firm can participate in a sequence of deviations by forming coalitions with workers across  periods. Each of these coalitions must be immediately profitable for the participating short-lived workers but not necessarily for the firm, since the firm cares about  the profit it collects from the entire sequence.

Motivated by this second observation, we define a \term{deviation plan} for firm $\firm$ as a complete contingent plan that specifies, at every ex post history, a set of workers to recruit and their wage offers. Formally, a deviation plan for firm $\firm$ is a pair $(d:\histories \rightarrow  2^{\workers},\eta: \histories \rightarrow \Re^{\abs{\workers}})$ such that $\left|d_{}(\history)\right| \le q_{\firm}$ for any $\history$ and $\eta_{ \worker}(\history) \ne 0$ only if $\worker\in d_{}(\history)$. Together with the original matching process, a deviation plan generates a distribution over the outcomes of the game $\fhistories_{\infty}$. Given a matching process $\mu$ and $\firm$'s deviation plan $(d_{}, \eta_{})$, the \term{manipulated matching process}, denoted by $[\mu, (\firm, d_{}, \eta_{})]: \histories \rightarrow \smatchings$, is a matching process defined by
\begin{equation*}
	\Big[\mu, (\firm, d_{}, \eta_{} )\Big] (\history) \equiv \Big [\mu(\history),\, \Big(\firm, d_{}(\history), \eta_{}(\history) \Big) \Big] \quad \forall h\in \histories.	
\end{equation*}

Firm $\firm$'s deviation plan $(d_{}, \eta_{})$ from $\mu$ is \term{feasible} if at every ex post history $h=(\overline{\history}, \type, \omega)$,
\begin{equation*}
	v_{\worker}\Big(\Big[\mu, (\firm, d_{}, \eta_{} )\Big] (\history),\, \type \Big) > v_{\worker} \Big( \mu(\history), \type \Big) \quad \forall \worker \in d_{} (\history).
\end{equation*} 
That is, workers participate in the deviation only if they find the new work environment strictly better than the recommendation from $\mu$.\footnote{An alternative formulation could allow workers who are already employed by $\firm$ to only weakly prefer their new work environment. Given the perfect divisibility of transfers, this would not affect our results.}
Lastly, the deviation plan $(d_{},\eta_{})$ is \term{profitable} if there exists an ex post history $\history$ such that $U_{\firm} \big(\history \,\big|\,[\mu, (\firm, d_{},\eta_{}) ]\big)> U_{\firm}(\history|\mu)$.

\paragraph{Self-Enforcing Matching Process.}
The two observations above motivate our notion of dynamic stability.
\begin{definition} \label{definition:self-enforcing-matching-process}
Matching process $\mu$ is \term{self-enforcing} if (i)
		$v_{\worker}(\mu(\history),\type) \geq 0 $ for every $\worker\in \workers$ at every ex post history $h\in \histories$; and
		(ii) no firm has a  feasible and profitable deviation plan.
\end{definition}

The first requirement guards against deviations by a single worker in any generation, and the second guards against deviations by firms. These requirements are imposed on the matching process at \emph{all} ex post histories, including those that are off path: This embeds a form of sequential rationality in the same way subgame perfection does in a repeated noncooperative game. It is also worth noting that \cref{definition:self-enforcing-matching-process} focuses on coalitions of a single firm and multiple workers: This is stronger than pairwise stability but weaker than group stability. In dynamic environments, these stability notions are not equivalent.\footnote{Allowing infinite-horizon deviations by multiple long-lived players creates a conceptual difficulty in assessing whether those deviations themselves can be self-enforcing. \cite{aliliu} show that if coalitions cannot commit to long-run behavior and are unable to make anonymous transfers, then modeling coalitions with multiple long-run players does not alter the set of sustainable outcomes when players are patient.} Nevertheless, in static settings, \cref{definition:self-enforcing-matching-process} reduces to the definition of core allocation in \cite{kelsocrawford1982}, so our definition can be viewed as its dynamic generalization.

\section{Results}\label{sec:results}

\paragraph{Minmax Payoffs.}
We begin by defining firms' minmax payoffs. For every realized type profile $\type \in \types$, let
\[
	M^{\circ}(\type) \equiv \{m\in M : v_{\worker}(m, \type) \geq 0 \text{ for all } \worker \in \workers\}
\]
denote the set of stage matchings that are individually rational for workers. From \cref{definition:self-enforcing-matching-process}, a self-enforcing matching process can only recommend stage matchings in $M^{\circ}(\type)$. Moreover, for every firm $\firm\in \firms$ and recommended stage matching $m = (\phi, {\wage})$, let
    \[
    D_{\firm}(m , \type) \equiv \left\{ (W', {\wage}'_{\firm}) : \begin{array}{c}
    |W'|\le q_{\firm}, \; \wage_{\firm \worker}'=0 \text{ if $\worker \notin W'$, }\\
    \tilde{v}_{\worker}(\firm , W' , \type)+\wage_{\firm \worker}'>v_{\worker}(m, \type) \text{ for every $\worker \in W'$}
    \end{array}\right\}
    \]  
denote the set of feasible stage-game deviations for $\firm$ at type profile $\type$. Firm $\firm$'s minmax payoff is its payoff from ``best responding'' to the worst recommendation.
\begin{definition}\label{def: minmax payoff}
    Firm $\firm$'s minmax payoff at type profile $\type$ is
    \[\underline{{u}}_{\firm}(\type) \equiv \inf_{m\in M^{\circ}(\type)}\, \sup_{ ( W', {\wage}'_{\firm})\in D_{\firm}(m, \type)} \,u_{\firm}([m, (\firm, W', {\wage}'_{\firm})], \type).\]
\end{definition}

To characterize this minmax payoff, let
\[
	\hchi({\firm},W, \type) \equiv \tilde{u}_{\firm}(W,\type) + \sum_{\worker \in W} \tilde{v}_{\worker}(\firm ,W\backslash\{\worker\},\type)
\]
denote the total surplus of coalition $(
\firm,W)$ at type profile $\type$.
The following lemma characterizes each firm's minmax payoff.
\begin{lemma}\label{lemma: equiv minmax}
    Let $Q\equiv\sum_{\firm'\in \firms} q_{\firm'}$ represent the sum of all firms' hiring quotas. For every firm $\firm$ and type profile $\type$, $\firm$'s minmax payoff satisfies
\begin{equation}\label{eq: minmax alt def}
    \underline{u}_{\firm}(\type) = \min_{W' \subseteq \workers, \abs{W'} \leq Q}\, \max_{ W \subseteq \workers\backslash W' , \abs{W}\leq q_{\firm}} \,\hchi({\firm},W, \type).
\end{equation}
\end{lemma}
\cref{lemma: equiv minmax} states that a firm's minmax payoff equals the maximum surplus it can generate after $\sum_{\firm'} q_{\firm'}$ workers have been removed from the market in an adversarial manner. The intuition is as follows. To penalize $\firm$, other firms can offer high wages to $\sum_{\firm'\ne f} q_{\firm'}$ workers, which renders them unattractive as potential partners for $\firm$. Also, if the employees of $\firm$ are paid high wages  from $\firm$ and require consent to join a deviation, then $\firm$ would be better off abandoning them and looking for cheaper workers for its deviation. Altogether, this excludes $\sum_{\firm'} q_{\firm'}$ workers from $\firm$'s potential match pool in case of a deviation. Instead, $\firm$ will select its partners from the remaining workers and secure all the matching surplus by offering those workers zero payoff.

\paragraph{Characterization.} We first introduce some notation. For every $\type\in \types$ and $m\in M^{\circ}(\type)$, let $u(m,\type)\equiv (u_{\firm}(m,\type))_{\firm\in \firms}$ denote firms' payoff profile under $m$. For every $\type\in\types$, let $$\mathcal{U}(\type)\equiv \conv{\,\big\{u\in \Re^{\firms}: u = u(m,\type) \text{ for some }m\in M^{\circ}(\type) \big\}\,}$$ denote the convex hull of these payoff profiles. In addition, define $$\mathcal{U}^*\equiv \Big\{\sum_{\type\in \types} \pi(\type) u(\type): u(\type)\in \mathcal{U}(\type) \text{ for every }\type \in \types \Big\}.$$ Finally, for every firm $\firm$, let 
\[
\underline{u}_{\firm}^* \equiv \exp_{\pi}[\underline{u}_{\firm}(\type)]
\]
denote its expected minmax payoff over type profiles. 

The next result characterizes the discounted payoffs from self-enforcing matching processes.

\begin{proposition} \label{prop: folk theorem} 
(i) If $u\in \mathcal{U}^*$ satisfies $u_{\firm} > \underline{u}_{\firm}^*$ for all $\firm\in \firms$, there is a $\underline{\delta}\in (0,1)$ such that for every $\delta\in(\underline{\delta},1)$, there exists a self-enforcing matching process with firms' discounted payoffs $u$.
(ii) Suppose $\mu$ is a self-enforcing matching process for a given $\delta\in (0,1)$. For every ex ante history $\overline{\history}\in \fhistories$, firms' discounted payoff profile satisfies $\big( U_{\firm}(\overline{\history}\,|\,\mu) \big)_{\firm\in \firms}\in \mathcal{U}^*$ and  $U_f(\overline{\history}\,|\,\mu) \geq \underline{u}^*_{\firm}$ for every $\firm\in \firms$.
\end{proposition}

To prove (i), we first show that $\mathcal{U}^*$ satisfies the NEU condition \citep{abreuduttasmith}. For every payoff profile in $\mathcal{U}^*$ that is strictly above every firm's average minmax payoff, this allows us to construct a set of payoff profiles, each individually tailored for a specific firm, which can be used as carrots and sticks to discipline the firms. Finally, these payoff profiles are integrated into the classical idea of \cite{fudenberg1986folk} to construct a matching process that is self-enforcing when firms are patient.

The cooperative nature of the stage game gives rise to a technical difficulty. Since stage matchings lack the usual product structure of strategy profiles in normal-form games, when a firm $\firm$ is being minmaxed, it might be necessary for $\firm$ to obtain stage payoffs that are even lower than its minmax value. When $\pi\in \Delta(\types)$ is not a degenerate distribution, this may create incentives for $\firm$ to deviate at some ex post histories when it is being minmaxed, even though it is not profitable to do so on average. We tackle this by adjusting wages in the punishment scheme so that the ex post benefits of deviations are equalized across realizations of $\type$ (see \cref{lemma: profit static minmax matchings}(iii)).

To prove (ii), note that given realized type profile $\type$, by the definition of minmax payoffs, every firm $\firm$ can secure the minmax payoff $\underline{u}_f(\type)$ by deviating with workers. Taking expectation over the distribution of type profiles delivers the result.

\cref{prop: folk theorem}, however, does not guarantee the existence of self-enforcing matching processes. In fact, statement (i) would be vacuously true if there is no payoff profile $u\in \mathcal{U}^*$ that satisfies $u_{\firm}>\underline{u}^*_f$ for all $\firm\in\firms$. Proving the existence of a self-enforcing matching process amounts to showing that such a payoff profile can arise from a random serial dictatorship among the firms.

\paragraph{No-Poaching Agreement.} 
Let $\Gamma\equiv \big\{\gamma:\firms \rightarrow \{1,\ldots,|\firms|\} \big\}$ denote the set of orderings over firms. In the serial dictatorship corresponding to an ordering $\gamma\in \Gamma$, firms sequentially choose workers based on $\gamma$ while setting all of their matched workers' payoffs to zero, thereby capturing the entire matching surplus.

\begin{definition}[Serial Dictatorship] \label{definition: serial dictatorship}
Given $\type\in \types$ and $\gamma\in \Gamma$, stage matching $\hat{m}(\type,\gamma)$ is induced by the following procedure. Initialize ${W}^{\#}_0\equiv\emptyset$.
\begin{itemize}[leftmargin=0.4cm]
    \item \textbf{Step $i=1, \ldots, |\firms|$}: Denote $\hat{\firm}_i \equiv \gamma^{-1}(i)$, and let $\hat{W}_i\in \argmax_{W\subseteq \workers\backslash{W}^{\#}_{i-1}, |W|\le q_{\hat{\firm}_i}} \hchi(\hat{\firm}_i,W,\type)$; \\
    Set $\phi(\hat{\firm}_i)= \hat{W}_i$, $p_{\hat{\firm}_i \worker}= -\tilde{v}_{\worker}(\hat{\firm}_i, \hat{W}_i\backslash\{\worker\})$ for every $\worker\in \hat{W}_i$, and $p_{\hat{\firm}_i \worker}=0$ for every $\worker\notin \hat{W}_i$; \\
    Update ${W}^{\#}_i\equiv{W}^{\#}_{i-1}\cup\hat{W}_i$.
\end{itemize}
\end{definition}
In a \emph{random serial dictatorship (RSD)}, firms randomize over $\Gamma$ and play matching $\hat{m}(\type,\gamma)$ based on the realized order $\gamma$. Note that RSD can be seen as a form of no-poaching agreement since a firm refrains from soliciting workers already employed by another firm even if it is feasible to do so.

We make an assumption to simplify our existence proof. Let
\begin{equation}\label{equation: firm max payoff}
    \overline{u}_\firm(\type) \equiv \max_{W \subseteq \workers,|W|\leq q_\firm}\,\hchi(\firm,W, \type)
\end{equation}
denote $\firm$'s maximum feasible payoff at type profile $\type$. It is easy to see that $\overline{u}_\firm(\type) \geq \underline{u}_{\firm}(\type)$ for every $\type\in \types$ and $\firm\in \firms$, since removing $Q$ workers in an adversarial manner reduces $\firm$'s maximum surplus. We assume that for every firm, this inequality is strict with a positive probability.
\begin{assumption}\label{assumption: generic}
    For every firm $\firm$, there exists $\type \in \types$ with $\pi(\type)>0$ such that $\overline{u}_\firm(\type) > \underline{u}_{\firm}(\type)$.
\end{assumption}
\cref{assumption: generic} holds generically, for example, when players' payoffs given every type profile $\type$ is randomly drawn from a continuous distribution. We make this assumption to simplify the exposition of \cref{prop: existence}. In Online Appendix we show that our results hold even without this assumption.

The following lemma shows that when firms randomize uniformly over serial dictatorships, they can simultaneously obtain payoffs that exceed their expected minmax payoffs.
\begin{lemma} \label{lemma: nonemptiness}
Under \cref{assumption: generic}, for every firm $\firm$,
\[
\frac{1}{\abs{\Gamma}}\sum_{\gamma\in \Gamma} \exp_{\pi}\Big[{u}_{\firm}\big(\,\hat{m}(\type,\gamma),\type\,\big)\Big] > \underline{u}_{\firm}^*.
\]
\end{lemma} 

Let us explain the intuition for \cref{lemma: nonemptiness}. Suppose firms randomize uniformly over $\Gamma$ for every realized type profile $\type$, then each firm $\firm$ receives the payoff on the left side.
Fix $\firm$ and $\type$. In this RSD, the worst-case scenario for $\firm$ arises when $\firm$ is ranked last in $\gamma$ (i.e.,  $\gamma(f)=|\firms|$): By the time $\firm$ selects its employees, $Q_{-\firm}\equiv \sum_{\firm'\ne \firm}q_{\firm'}$ workers are already off the market. However, for $\firm$, this worst case under RSD is still weakly more desirable than being minmaxed, in which case $\firm$ must choose its employees after $Q=\sum_{\firm'}q_{\firm'}$ workers have been eliminated in an \textit{adversarial} manner. Therefore, the distribution of payoffs $\firm$ obtained under RSD first-order stochastically dominates the distribution of its minmax payoffs, so in expectation, every firm prefers RSD over being minmaxed.

We now state our existence result.

\begin{proposition} \label{prop: existence}
When firms are sufficiently patient, there exists a self-enforcing matching process in which firms play RSD in every period on path.
\end{proposition} 

This result shows that as a particular form of no-poaching agreement, RSD can always be sustained dynamically. To understand the intuition, consider the uniform RSD. For every $\type$ and $\gamma$, matching $\hat{m}(\type,\gamma)$ is in $M^{\circ}(\type)$, so the randomization over $\Gamma$ places firms' payoff profile in $\mathcal{U}(\type)$. Taking expectation over $\type$, we know that firms' payoff profile from the uniform RSD is in $\mathcal{U}^*$. Furthermore, according to \cref{lemma: nonemptiness}, each firm's payoff is strictly higher than its average minmax payoff. \cref{prop: folk theorem}(i) then delivers the result.

\appendix
\renewcommand{\theequation}{\arabic{equation}}
  \renewcommand{\thesection}{\Alph{section}}

\section{Appendix} \label{section: appendix}

\subsection{Intermediate Results}\label{appendix: preliminary lemmas}

\begin{proofof}{\cref{lemma: equiv minmax}}
Fix $\firm\in\firms$, $\type\in\types$, and stage matching $m=(\phi, {\wage})\in M^{\circ}(\type)$. Define $W_m\equiv\{\worker \in \workers: \phi(\worker) \neq (\emptyset, \emptyset)\}$, so $\abs{W_m}\leq Q$. Hiring from $\workers\backslash W_m$ at infinitesimal wages is always a feasible deviation for $\firm$, which means
\begin{align*}
    \sup_{ (W', {\wage}'_{\firm})\in D_{\firm}(m, \type)} \,u_{\firm}([m, (\firm, W', {\wage}'_{\firm})], \type) & \geq \max_{ W' \subseteq \workers\backslash W_m , \abs{W'}\leq q_{\firm}} \,\hchi({\firm},W', \type)\\
    & \geq \min_{W \subseteq \workers, \abs{W} = \abs{W_m}}\, \max_{ W' \subseteq \workers\backslash W , \abs{W'}\leq q_{\firm}} \,\hchi({\firm},W', \type)\\
    & \geq \min_{W \subseteq \workers, \abs{W} \le Q}\, \max_{ W' \subseteq \workers\backslash W, \abs{W'}\leq q_{\firm}} \,\hchi({\firm},W', \type).
\end{align*}
Taking infimum on the LHS yields $$\underline{{u}}_{\firm}(\type) \geq \min_{W \subseteq \workers, \abs{W} \le Q}\, \max_{ W' \subseteq \workers\backslash W , \abs{W'}\leq q_{\firm}} \,\hchi({\firm},W', \type).$$

For the other direction, take any $W\subseteq \workers$ with $\abs{W} \le Q$. We can always construct a stage matching $\hat{m}=(\hat{\phi}, \hat{{\wage}})\in M^{\circ}(\type)$ such that (i) $\hat{\phi}$ assigns all workers in $W$ to firms, and (ii) all workers in $W$ receive sufficiently high wages so that $\firm$ never finds it profitable to deviate with any workers in $W$.
Therefore,
\begin{align*}
    \max_{ W' \subseteq \workers\backslash W , \abs{W'}\leq q_{\firm}} \,\hchi({\firm},W', \type) & = \sup_{ (W', {\wage}'_{\firm})\in D_{\firm}(\hat{m}, \type)} \,u_{\firm}([\hat{m}, (\firm, W', {\wage}'_{\firm})], \type)\\
    & \geq \inf_{m\in M^{\circ}(\type)}\, \sup_{ (W', {\wage}'_{\firm})\in D_{\firm}(m, \type)} \,u_{\firm}([m, (\firm, W', {\wage}'_{\firm})], \type)\end{align*}
Minimizing over $W$ on the LHS yields the other direction.
\end{proofof}

\begin{lemma}
\label{lemma: one-shot-deviation}
Deviation plan $(d_{}, \eta_{})$ is a \textup{one-shot deviation} from matching process $\mu$ if there is a unique ex post history $\hat{\history}$ where $[\mu, (\firm, d_{},\eta_{}) ](\hat{\history}) \ne \mu (\hat{\history})$. Matching process $\mu$ is self-enforcing if and only if (i) $v_{\worker}(\mu(\history), \type) \geq 0 $ for every $\worker\in \workers$ at every $\history\in \histories$; and (ii) no firm has a feasible and profitable \textup{one-shot deviation}.
\end{lemma}

\begin{proof}
Standard arguments \citep{blackwell1965discounted} suffice. Online Appendix contains the proof.
\end{proof}

\begin{lemma}
\label{lemma: identifiability of manipulator} For any stage matching $m$, $$[m, (\firm_1, W_{\firm_1}, {\wage}_{\firm_1})]=[m, (\firm_2, W_{\firm_2}, {\wage}_{\firm_2})]\ne m  \text{ implies } \firm_1 = \firm_2.$$
\end{lemma}

\begin{proof}
Let $m = (\phi, {\wage})$, $[m, (\firm_1, W'_{\firm_1}, {\wage}'_{\firm_1})] = (\overline{\phi}, \overline{{\wage}})$, and $[m, (\firm_2, W'_{\firm_2}, {\wage}'_{\firm_2})] = (\hat{\phi}, \hat{{\wage}})$. Toward contradiction, suppose $\firm_1 \ne \firm_2$, but $[m, (\firm_1, W'_{\firm_1}, {\wage}'_{\firm_1})] = [m, (\firm_2, W'_{\firm_2}, {\wage}'_{\firm_2})] \ne m $. Then $\overline{\phi} =\hat{\phi}$ and $\overline{{\wage}}_{\firm}=\hat{{\wage}}_{\firm}$ for every $\firm\in \firms$. Each of the three cases to consider yields a contradiction.
\begin{enumerate}
	\item Suppose $W'_{\firm_1} = \phi (\firm_1)$ and $W'_{\firm_2} = \phi(\firm_2)$. Since $[m, (\firm_1, W'_{\firm_1}, {\wage}'_{\firm_1})] \ne m $,  we have ${\wage}'_{\firm_1} \ne {\wage}_{\firm_1}$. Then $\overline{{\wage}}_{\firm_1} = {\wage}'_{\firm_1} \ne {\wage}_{\firm_1} = \hat{{\wage}}_{\firm_1} $, a contradiction.
	\item Suppose $W'_{\firm_1} \subseteq \phi (\firm_1)$ and $W'_{\firm_2} \subseteq \phi(\firm_2))$ but, without loss of generality, $W'_{\firm_1} \neq \phi (\firm_1)$. We have $\overline{\phi}(\firm_1)  = W'_{\firm_1} \ne \phi(\firm_1) = \hat{\phi}(\firm_1) $, so $\overline{\phi} \ne \hat{\phi}$, a contradiction.
	\item Suppose, without loss of generality, $W'_{\firm_1} \nsubseteq \phi (\firm_1)$. Let $\worker' \in W'_{\firm_1} \backslash \phi (\firm_1)$. We have $\overline{\phi} (\worker') = \firm_1$, whereas $\hat{\phi} (\worker') \in \{\phi (\worker'), \firm_2 \}$. Since $\worker' \notin \phi (\firm_1)$ and $\firm_1 \ne \firm_2$, $\firm_1 \notin \{\phi (\worker'), \firm_2 \}$. Hence, $\overline{\phi} \ne \hat{\phi} $, a contradiction.
\end{enumerate}
Therefore, $[m, (\firm_1, W_{\firm_1}, {\wage}_{\firm_1})] = [m, (\firm_2, W_{\firm_2}, {\wage}_{\firm_2})] \ne m $ implies $\firm_1 = \firm_2$.
\end{proof}

\subsection{Proof of \texorpdfstring{\cref{prop: folk theorem}}{Proposition 1}}\label{appendix: prop folk}

\begin{lemma} \label{lemma: profit static minmax matchings}
For each $\type\in \types$, there exist stage matchings $\big\{\underline{m}_\firm(\type)\big\}_{\firm \in \firms} \subseteq M^{\circ}(\type)$ such that $\forall \firm\in \firms$, 
\begin{enumerate}
\item[(i)] $\sup_{(W', {\wage}'_\firm) \in D_\firm(\underline{m} _\firm, \type) } u_\firm([\underline{m}_\firm(\type), (\firm, W', {\wage}'_\firm)], \type ) = \underline{u}_\firm(\type)$;
\item[(ii)] $ u_\firm(\underline{m}_\firm(\type), \type) \le \underline{u}_\firm(\type)$; 
\item[(iii)] $\underline{u}_\firm(\type)- u_\firm(\underline{m}_\firm(\type), \type)= \underline{u}_\firm^*- \exp_{\pi}[u_\firm(\underline{m}_\firm(\type), \type)]$.
\end{enumerate}
\end{lemma}

\begin{proof}
Fix $\type \in \types$. For every $\firm \in \firms$, let
\[
	\underline{W}_{\firm}(\type) \in \argmin_{W \subseteq \workers, \abs{W} \leq Q}\, \max_{ W' \subseteq \workers\backslash W , \abs{W'}\leq q_\firm} \,\hchi(\firm,W', \type)
\]
denote a set of workers to eliminate that minmaxes $\firm$,
and
\[
b_\firm(\type) = \overline{u}_\firm(\type) - \underline{u}_\firm(\type)\ge 0
\]
denotes the difference between $\firm$'s maximum feasible payoff and minmax payoff.

For every $\firm\in\firms$, let $\big\{\underline{W}_{\firm}^{\firm'}(\type) \big\}_{\firm' \in \firms}$ be a partition of $\underline{W}_{\firm}(\type)$ such that $\abs{\underline{W}_{\firm}^{\firm}(\type)} \geq 1$ and $\abs{\underline{W}_{\firm}^{\firm'}(\type)} \leq q_{\firm'}$ for every $\firm'\in \firms$. Define
\[
B_f\equiv\max_{\type\in\types}\max_{\worker\in\underline{W}_\firm^\firm(\type)}\,\abs{\underline{W}_{\firm}^{\firm}(\type)}\Big[b_\firm(\type)-\tilde{v}_\worker \big(\firm, \underline{W}_{\firm}^{\firm}(\type)\backslash\{\worker\} , \type \big)\Big]+\underline{u}_\firm(\type)-\tilde{u}_\firm(\underline{W}_{\firm}^{\firm}(\type), \type).
\]
For every $\type \in \types$, define stage matching $\underline{m}_\firm(\type) = (\underline{\phi}^\firm(\type), \underline{{\wage}}^\firm(\type) )$, where
\[
	\underline{\phi}^\firm(\type)(\firm') = \underline{W}_{\firm}^{\firm'}(\type),
\]
and
\[
	\underline{p}_{\firm'\worker}^\firm (\type) = 
	\begin{cases}
        \frac{\tilde{u}_\firm(\underline{W}_{\firm}^{\firm}(\type), \type) -\underline{u}_\firm(\type) +B_\firm}{\abs{\underline{W}_{\firm}^{\firm}(\type)}} & \text{if } \firm' = \firm \text{ and } \worker \in \underline{W}_{\firm}^{\firm}(\type)\\
	b_\firm(\type) - \tilde{v}_\worker \big(\firm', \underline{W}_{\firm}^{\firm'}(\type)\backslash\{\worker\} , \type \big) & \text{if } \firm' \neq \firm \text{ and } \worker \in \underline{W}_{\firm}^{\firm'}(\type) \neq \emptyset\\
	0 & \text{otherwise}.
	\end{cases}
\]	

For any $\type \in \types$, if $\worker \in \underline{W}_{\firm}^{\firm}(\type)$, 
\[
v_\worker(\underline{m}_\firm(\type), \type) = \tilde{v}_\worker \big(\firm, \underline{W}_{\firm}^{\firm}(\type)\backslash\{\worker\}, \type \big)+\frac{\tilde{u}_\firm(\underline{W}_{\firm}^{\firm}(\type), \type) -\underline{u}_\firm(\type) +B_\firm}{\abs{\underline{W}_{\firm}^{\firm}(\type)}} \geq b_\firm(\type) \geq  0
\]
by the definition of $B_\firm$. If $\worker \in \underline{W}_{\firm}^{\firm'}(\type)$ and $\firm'\neq \firm$, \[
v_\worker(\underline{m}_\firm(\type), \type)=b_\firm(\type) \geq 0.
\]
Hence, $\underline{m}_\firm(\type) \in M^{\circ}(\type)$.

\begin{enumerate}
\item[(i)] 
Consider any feasible deviation $( W', {\wage}'_\firm) \in D_\firm(\underline{m} _\firm, \type)$. Suppose $W' \subseteq \workers \backslash \underline{W}_{\firm}(\type)$. By feasibility, every $\worker \in W'$ finds the deviation individually rational, so $$\tilde{v}_\worker(\firm, W'\backslash\{\worker\}, \type) + \wage'_{\firm \worker} \ge 0.$$ This implies 
\begin{eqnarray}
\tilde{u}_\firm(W', \type) - \sum_{\worker \in W'} \wage'_{\firm \worker} & \le & \tilde{u}_\firm(W', \type) + \sum_{\worker \in W'} \tilde{v}_\worker(\firm, W'\backslash\{\worker\}, \type) \label{firstineq} \\
& = & \hchi(\firm, W' , \type) \nonumber \\
& \le & \max_{ W \subseteq \workers\backslash \underline{W}_{\firm}(\type) , \abs{W}\leq q_\firm} \,\hchi(\firm, W, \type) \label{lastineq} \\
& = & \underline{u}_\firm(\type), \nonumber
\end{eqnarray}
where \eqref{firstineq} follows from $W' \subseteq \workers \backslash \underline{W}_{\firm}(\type)$, and \eqref{lastineq} follows from the definition of $\underline{W}_{\firm}(\type)$. Suppose instead $W' \nsubseteq \workers \backslash \underline{W}_{\firm}(\type)$. Fix $\worker \in W' \cap \underline{W}_{\firm}(\type)$. By the construction of $\underline{m}_\firm(\type)$, $v_{\worker}(\underline{m}_\firm(\type), \type) \geq b_\firm(\type)$. For the deviation to be feasible, $\worker$ must obtain a payoff weakly higher than in $\underline{m}_\firm(\type)$
: $$\tilde{v}_\worker(\firm, W'\backslash\{\worker\}, \type) + \wage'_{\firm \worker} \ge b_\firm(\type);$$ meanwhile, every $\worker' \in W' \backslash \{\worker\}$ needs to find the deviation individually rational: $$\tilde{v}_{\worker'}(\firm, W'\backslash\{\worker'\}, \type) + \wage'_{\firm \worker'} \ge 0.$$ Then
\begin{align*}
\tilde{u}_\firm(W', \type) - \sum_{\worker' \in W'} \wage'_{\firm \worker'} & = \tilde{u}_\firm(W', \type) - \wage'_{\firm \worker} - \sum_{\worker' \in W', \worker' \ne \worker} \wage'_{\firm \worker'} \\
& \le \tilde{u}_\firm( W', \type) + \sum_{ \worker\in W' } \tilde{v}_\worker (\firm, W'\backslash\{\worker\}, \type ) -  b_\firm(\type) \\
& = \hchi(\firm, W', \type) - \overline{u}_\firm(\type) + \underline{u}_\firm(\type) \\
& \le \underline{u}_\firm(\type).
\end{align*}
We have shown that for any $W'$, $$ u_\firm([\underline{m} _\firm, (\firm, W', {\wage}'_\firm)], \type ) \le \underline{u}_\firm(\type). $$
Hence,
\[
\sup_{( W', {\wage}'_\firm) \in D_\firm(\underline{m} _\firm, \type) } u_\firm([\underline{m} _\firm, (\firm, W', {\wage}'_\firm)], \type ) \le \underline{u}_\firm(\type).
\]
By the definition of $\underline{u}_\firm(\type)$, the above holds with equality.

\item[(ii)] For every $\firm \in \firms$ and $\type \in \types$,
\begin{eqnarray}
u_\firm(\underline{m}_\firm(\type), \type) & = & \tilde{u}_\firm( \underline{W}^\firm_\firm(\type), \type) - \sum_{ \worker\in \underline{W}^\firm_\firm(\type) } \underline{p}_{\firm \worker}^\firm(\type)	 \nonumber\\
& =& \tilde{u}_\firm( \underline{W}^\firm_\firm(\type), \type) - \tilde{u}_\firm(\underline{W}_{\firm}^{\firm}(\type), \type) + \underline{u}_\firm(\type) -B_\firm \nonumber \\
& =& \underline{u}_\firm(\type) -B_\firm \label{uf-Bf}  \\
& \leq& \tilde{u}_\firm( \underline{W}^\firm_\firm(\type), \type) + \sum_{ \worker\in \underline{W}^\firm_\firm(\type) } \tilde{v}_i (\firm, \underline{W}^\firm_\firm(\type)\backslash\{\worker\}, \type ) - \abs{\underline{W}_{\firm}^{\firm}(\type)} b_\firm(\type) \label{ineq1} \\
& \le & \tilde{u}_\firm( \underline{W}^\firm_\firm(\type), \type) + \sum_{ \worker\in \underline{W}^\firm_\firm(\type) } \tilde{v}_i (\firm, \underline{W}^\firm_\firm(\type)\backslash\{\worker\}, \type ) - b_\firm(\type) \label{ineq2} \\
& =& \hchi(\firm, \underline{W}^\firm_\firm(\type), \type ) - \overline{u}_\firm(\type) + \underline{u}_\firm(\type) \nonumber\\
& \le & \underline{u}_\firm(\type) \nonumber,
\end{eqnarray}
where \eqref{ineq1} follows from the definition of $B_\firm$, and \eqref{ineq2} follows from $\abs{\underline{W}_{\firm}^{\firm}(\type)} \ge 1$.
 
\item[(iii)] By \eqref{uf-Bf}, for every $\firm\in \firms$ and $\type \in \types$,
\[
u_\firm(\underline{m}_\firm(\type), \type)= \underline{u}_\firm(\type) - B_\firm. 
\]
Therefore, 
\[
\exp_{\pi}[u_\firm(\underline{m}_\firm(\type), \type)] =\underline{u}^*_\firm - B_\firm,
\]
which means
\[
\underline{u}_\firm(\type) - u_\firm(\underline{m}_\firm(\type), \type)=\underline{u}_\firm^* - \exp_{\pi}[u_\firm(\underline{m}_\firm(\type), \type)].
\]
\end{enumerate}
\end{proof}

\begin{lemma} \label{lemma: fw firm-specific punishments}
For every $u \in \mathcal{U}^{*}$, there exist vectors $\{ u^\firm\}_{\firm \in \firms} \subseteq \mathcal{U}^{*}$ such that for all distinct $ \firm,\firm'\in \firms$,
\[
	u^\firm_\firm < u_\firm 
\quad \text{and} \quad
	u^\firm_\firm < u^{\firm'}_\firm.
\]
\end{lemma}

\begin{proof}
Take any distinct $\firm,\firm'\in \firms$. If $\firm$ pays higher wages to its employees, $\firm$ is worse off while $\firm'$ is indifferent. $\mathcal{U}^*$ satisfies the NEU condition \citep{abreuduttasmith}, and $\{ u^\firm\}_{\firm \in \firms}$ can be obtained via a similar construction. See Online Appendix.
\end{proof}

\begin{proofof}{\cref{prop: folk theorem}(i)}
Fix $u\in \mathcal{U}^*$. Let $(\lambda(\type))_{\type\in \types}$ be a tuple of lotteries such that $\lambda(\type)\in \Delta(M^{\circ}(\type)) \;\forall\type\in \types$, and $u = \exp_{\pi}[\exp_{\lambda(\type)}[u (m, \type) ]]$.
Let $\{\underline{m}_\firm(\type)\}_{\type\in \types,\firm\in \firms}$ be the minmax stage matchings constructed in \cref{lemma: profit static minmax matchings}. 
By \cref{lemma: profit static minmax matchings}(ii),
\begin{equation} \label{equation: profit minmax match lower than minmax}
	u_\firm( \underline{m}_\firm(\type), \type) \leq \underline{u}_\firm(\type) \quad \forall (\firm,\theta)\in \firms\times \Theta.
\end{equation}
By \cref{lemma: fw firm-specific punishments}, there exist $\{u^{\firm}\}_{\firm \in \firms} \subseteq \mathcal{U}^{*}$ such that for all distinct $\firm,\firm'\in \firms$,
\[
	u^\firm_\firm < u_\firm 
\quad \text{and} \quad
	u^\firm_\firm < u^{\firm'}_\firm.
\]  For every $\firm \in \firms$, let $(\lambda^\firm(\type))_{\type \in \types}$ be the tuple of lotteries that satisfies $\lambda^f(\type)\in \Delta(M^{\circ}(\type)) \;\forall\type\in \types$, and $u^f = \exp_{\pi}[\exp_{\lambda^f(\type)}[u (m, \type) ]]$. By Carath\'{e}odory's theorem, it is without loss to assume every $\lambda^\firm(\type)$ has bounded support. Define
\begin{equation}\label{eq: def of Z}
    Z_\firm(\type)\equiv\inf_{m\in\supp (\lambda^\firm(\type))}\, u_\firm(m, \type).
\end{equation}

At any $\type$, let $\overline{u}_\firm(\type)$ denote a firm $\firm$'s maximum payoff from any matching in $M^{\circ}(\type)$ as defined in \eqref{equation: firm max payoff}. Note that $\overline{u}_\firm(\type)$ is an upper bound for the payoff $\firm$ can obtain in any feasible deviation from matchings in $M^{\circ}(\type)$. Let $L$ be an integer greater than
\[
    \max_{\firm \in \firms,\, \type \in \types}\,\frac{\overline{u}_\firm(\type) - Z_\firm(\type)}{u^\firm_\firm-\underline{u}^*_\firm}.
\]
Consider a matching process with the following three phases:
\begin{itemize}
    \item[(I)] If past realizations from $\lambda(\cdot)$ were followed: Match according to $\lambda(\cdot)$;
    \item[(II)] If $\firm$ deviates: For the next $L$ periods, match according to $\underline{m}_\firm(\cdot)$;  \item[(III)] If $\underline{m}_\firm(\cdot)$ was followed for $L$ periods:  Match according to the realization from $\lambda^\firm(\cdot)$ until a firm deviates.
\end{itemize}
Note that if firm $\firm'$ deviates from (II) or (III), the process restarts (II) with $\firm$ replaced by $\firm'$. All deviations by individual workers are ignored.
\smallskip

By \cref{lemma: identifiability of manipulator}, for any stage matching that can result from a firm's deviation, we can uniquely identify the firm, so the transition above is well-defined.
All stage matchings in this matching process are individually rational for workers. It remains to check that no firm has profitable one-shot deviations.

\medskip
\noindent \underline{(I)}
$\firm$ has no profitable one-shot deviation if
\[
    ( 1 - \delta ) u_{\firm}(m, \type) + \delta u_{\firm} \ge (1-\delta) \overline{u}_{\firm}(\type) +\delta(1-\delta^{L}) \sum_{\type'\in \types} \pi[\type']u_{\firm}(\underline{m}_{\firm}(\type'), \type') + \delta^{L +1} u^{\firm}_{\firm}.
\]
Since $u_{\firm}  > u^{\firm}_{\firm}$ by construction, $\firm$ has no profitable one-shot deviation for $\delta$ high enough.

\medskip
\noindent \underline{(II)} Consider two cases.

\underline{Case a:} $\firm' \neq \firm$. Without deviation, $\firm'$ gets
\begin{equation}\label{equation: profit other firm minmax no dev}
    (1-\delta) u_{\firm'}(\underline{m}_\firm(\type) , \type) + \delta (1-\delta^{L - \tau-1})  \sum_{\type'\in \types} \pi[\type']u_{\firm'}(\underline{m}_\firm(\type') , \type') + \delta^{L - \tau} u^\firm_{\firm'}.
\end{equation}
By deviating, $\firm'$ gets less than
\begin{equation}\label{equation: profit other firm minmax dev}
    (1-\delta) \overline{u}_{\firm'}(\type) +\delta(1-\delta^{L}) \sum_{\type'\in \types} \pi[\type']u_{\firm'}(\underline{m}_{\firm'}(\type'), \type') + \delta^{L+1} u^{\firm'}_{\firm'}.
\end{equation}
As $\delta \rightarrow 1$, (\ref{equation: profit other firm minmax no dev}) converges to $u^\firm_{\firm'}$ and (\ref{equation: profit other firm minmax dev}) to $u^{\firm'}_{\firm'}$. Because $u^\firm_{\firm'}>u^{\firm'}_{\firm'}$ by construction, no one-shot deviation is profitable for high $\delta$.

\underline{Case b:} $\firm' = \firm$. 
Without deviation, $\firm$ gets
\begin{equation}\label{equation: profit same firm minmax no dev}
 (1-\delta) u_{\firm}(\underline{m}_\firm(\type) , \type) + \delta (1-\delta^{L - \tau-1})  \sum_{\type'\in \types} \pi[\type']u_{\firm}(\underline{m}_\firm(\type') , \type') + \delta^{L - \tau} u^\firm_{\firm}.
\end{equation}
With deviation, $\firm$ gets less than
\begin{equation}\label{equation: profit same firm minmax dev}
(1-\delta) \underline{u}_\firm(\type) +\delta(1-\delta^{L}) \sum_{\type'\in \types} \pi[\type']u_{\firm}(\underline{m}_{\firm}(\type'), \type') + \delta^{L+1} u^{\firm}_{\firm}.
\end{equation}
Since $ \sum_{\type'\in \types} \pi[\type']u_{\firm}(\underline{m}_\firm(\type') , \type')  \le \underline{u}^*_\firm < u^\firm_\firm$, expression \eqref{equation: profit same firm minmax no dev} is increasing in $\tau$. Hence it suffices to prove the case when $\tau=0$. When $\tau=0$, 
\begin{equation}\label{inequality above}
	\left[ u_{\firm}(\underline{m}_\firm(\type) , \type) - \underline{u}_\firm(\type)\right] -\delta^{L}\sum_{\type'\in \types} \pi[\type']u_{\firm}(\underline{m}_\firm(\type') , \type') + \delta^{L}u^\firm_\firm \geq 0.
\end{equation}
By \cref{lemma: profit static minmax matchings}(iii), 
\begin{equation}\label{equality above}
    u_{\firm}(\underline{m}_\firm(\type) , \type) - \underline{u}_\firm(\type) = \sum_{\type'\in \types} \pi[\type']u_{\firm}(\underline{m}_\firm(\type') , \type')  - \underline{u}^*_\firm.
\end{equation}
Substituting \eqref{equality above} into \eqref{inequality above} yields
\[
    (1-\delta^{L}) \sum_{\type'\in \types} \pi[\type']u_{\firm}(\underline{m}_\firm(\type') , \type') + \delta^{L}u^\firm_\firm \geq \underline{u}^*_\firm.
\]
By construction, $u^\firm_\firm>\underline{u}^*_\firm$, so the inequality holds for $\delta$ high enough.

\medskip
\noindent \underline{(III)} Consider two cases.

\underline{Case a:} $\firm' \neq \firm$. The argument for (I) holds analogously here, once we replace $\firm$ with $\firm'$ and $u_\firm$ with $u^\firm_{\firm'}$.

\underline{Case b:} $\firm' = \firm$.
There is no profitable one-shot deviation for $\firm'$ if
\[
	( 1 - \delta ) u_{\firm'}(m, \type) + \delta u_{\firm'}^{\firm'}  \ge (1-\delta) \overline{u}_{\firm'}(\type) +\delta(1-\delta^{L}) \sum_{\type'\in \types} \pi[\type']u_{\firm'}(\underline{m}_{\firm'}(\type'), \type') + \delta^{L+1} u^{\firm'}_{\firm'}.
\]
This is equivalent to
\[
    \delta \frac{1- \delta^{L}}{1- \delta} \left[u^{\firm'}_{\firm'} - \sum_{\type'\in \types} \pi[\type']u_{\firm'}(\underline{m}_{\firm'}(\type'), \type') \right] \ge \overline{u}_{\firm'}(\type) - u_{\firm'}(m, \type).
\]
Because $\frac{1- \delta^{L}}{1- \delta}\rightarrow L$ as $\delta \rightarrow 1$, the LHS converges to 
\begin{eqnarray}
    && L\left[u^{\firm'}_{\firm'} - \sum_{\type'\in \types} \pi[\type']u_{\firm'}(\underline{m}_{\firm'}(\type'), \type') \right] \nonumber \\
    &> & \frac{\overline{u}_{\firm'}(\type) - Z_{\firm'}(\type)}{u^{\firm'}_{\firm'}-\underline{u}^*_{\firm'}}\left[u^{\firm'}_{\firm'} - \sum_{\type'\in \types} \pi[\type']u_{\firm'}(\underline{m}_{\firm'}(\type'), \type') \right] \label{ineqfirst}\\
    & \ge & \overline{u}_{\firm'}(\type) - u_{\firm'}(m, \type), \label{ineqsec}
\end{eqnarray}
where \eqref{ineqfirst} follows from the definition of $L$, and \eqref{ineqsec} follows from \eqref{equation: profit minmax match lower than minmax}, \eqref{eq: def of Z}, and $m\in\supp (\lambda^\firm(\type))$. Thus, no deviation is profitable for $\delta$ high enough.
\end{proofof}

\begin{proofof}{\cref{prop: folk theorem}(ii)}
By definition, every firm $\firm$ can secure $\underline{u}_f(\type)$ by deviating with workers. Taking expectation over $\type$ delivers the result. See Online Appendix.
\end{proofof}

\subsection{Proof of \texorpdfstring{\cref{lemma: nonemptiness}}{Lemma 2}}\label{appendix: prop nonemptiness}

For every $\firm \in \firms$, let $\Gamma(\firm)  = \{\gamma \in \Gamma : \gamma(\firm) = 1 \}$ denote the set of orderings in which $\firm$ is ranked first.

Recall that for every $\gamma \in \Gamma$ and $\type \in \types$, matching $\hat{m}(\type,\gamma)$ is produced by serial dictatorship (\cref{definition: serial dictatorship}). It is straightforward that
\begin{equation}\label{equation: serial dictatorship 1}
u_{\firm}(\hat{m}(\theta,\gamma) ) =  \overline{u}_\firm(\type).
\end{equation}
Additionally, in serial dictatorship, for each $\hat{m}(\type,\gamma)$ and $\firm\in \firms$,
\begin{equation*}
    u_{\firm}(\hat{m}(\type,\gamma)) = u_{\firm}(\hat{W}_{\gamma(f)}) = \max_{W'\subseteq \workers \backslash \cup_{1\le i<\gamma(\firm)} \hat{W}_{i}, \abs{W'}\leq q_\firm} \, \hchi(\firm, W', \type).
\end{equation*}
Since every $\gamma \in \Gamma$ satisfies $\abs{\cup_{1\leq  i< \gamma(\firm)} \hat{W}_{i}} \le Q$,
\begin{equation} \label{equation: serial dictatorship 2}
	u_{\firm}(\hat{m}(\type,\gamma)) \ge \min_{W \subseteq \workers, \abs{W} \leq Q}\, \max_{ W' \subseteq \workers\backslash W , \abs{W'}\leq q_\firm } \,\hchi(\firm, W', \type) = \underline{u}_\firm(\type) \quad \forall \type\in \types, \firm\in \firms, \gamma\in \Gamma.
\end{equation}
For every $\firm\in \firms$,
\begin{align*}
    \frac{1}{|\Gamma|}\sum_{\gamma\in \Gamma} \exp_{\pi}\Big[{u}_{\firm}\big(\,\hat{m}(\type,\gamma),\type\,\big)\Big] & = \sum_{\type\in \types} \pi(\type)
    \sum_{\gamma\in \Gamma} \frac{1}{|\Gamma|}\Big[{u}_{\firm}\big(\,\hat{m}(\type,\gamma),\type\,\big)\Big] 
     > \sum_{\type\in \types} \pi(\type) \,\underline{u}_{\firm}(\type) = \underline{u}_{\firm'}^*,
\end{align*}
where the strict inequality follows from \eqref{equation: serial dictatorship 1}, \eqref{equation: serial dictatorship 2}, and \cref{assumption: generic}. 
 
\newpage
{\small
	\addcontentsline{toc}{section}{References}
	\setlength{\bibsep}{0.25\baselineskip}
	\bibliographystyle{aer}
	\bibliography{DynMatch}
}
\newpage

\newpage
\section{Online Appendix (For Online Publication)} \label{appendix: online}

\subsection{Omitted Proofs}
\begin{proof}[\bf Proof of \cref{lemma: one-shot-deviation}]
Suppose one-shot deviation plan $(d_{}, \eta_{})$ from matching process $\mu$ for firm $\firm$ is feasible and profitable. Since stage payoffs are bounded for firm $\firm$ and there is discounting, the standard one-shot deviation principle for individual decision-making \citep*{blackwell1965discounted} implies that there exists an ex post history $\hat{\history}=(\hat{\overline{\history}}, \hat{\type}, \hat{\omega})$ such that
\[(1-\delta)\left[\tilde{u}_{\firm}(d_j(\hat{\history}), \hat{\type})-\sum_{\worker\in d_{}(\hat{\history})}\eta_{ \worker}(\hat{\history})\right] + \delta U_{\firm}\big(\big[\mu(\hat{\history}),(\firm, d_{}(\hat{\history}), \eta_{}(\hat{\history}) ) \big]\, \big | \, \mu \big) > U_{\firm}(\hat{\history}\,|\, \mu ).\]

Consider deviation plan $(d^o_{\firm},\eta^o_{\firm})$ that satisfies 
\begin{equation*}
    [\mu,(\firm,d^o_{\firm},\eta^o_{\firm})](\history) = \begin{cases}
    [\mu,(\firm,d_{},\eta_{})](\history) &\text{if $h =\hat{\history}$,}\\
    \mu(\history) &\text{otherwise.}
    \end{cases}
\end{equation*}
Then $(d^o_{\firm},\eta^o_{\firm})$ is a profitable one-shot deviation plan for firm $\firm$.
\end{proof}

\begin{proof}[\bf Proof of \cref{lemma: fw firm-specific punishments}]

Fix stage matching $m = (\phi, {\wage} )$ and type profile $\type$. For every $\firm \in \firms$, define a matching $m^\firm = (\phi^\firm, {\wage}^\firm )$ by $\phi^\firm = \phi$, and
\[
	\wage^\firm_{\firm'\worker} = 
	\begin{cases}
		\wage_{\firm'\worker} + 1 & \text{ if } \firm' = \firm \text{ and } \worker \in \phi^\firm, \\
		\wage_{\firm'\worker}  & \text{ otherwise}.\\
	\end{cases}
\]
Clearly, $v_{\worker}(m^\firm, \type) = v_{\worker}(m, \type)$ for every $\worker \in \workers \backslash \phi^\firm$, and $v_{\worker}(m^\firm, \type) > v_{\worker}(m, \type)$ for every $\worker \in \phi^\firm$, so $m^\firm$ is individually rational for all workers whenever $m$ is---i.e., $m \in M^{\circ}(\type)$ implies $m^\firm \in M^{\circ}(\type)$. Moreover, by the construction of $\{m^\firm\}_{\firm \in \firms}$, we have $u_\firm(m^\firm, \type) \le u_\firm(m, \type)$ for every $\firm\in \firms$, with the inequality being strict if $\phi(\firm) \neq \emptyset$. In the following, we write $\zeta^\firm: m \mapsto m^\firm$.

For an arbitrary vector $u \in \mathcal{U}^{*}$, there exists a collection $(\lambda(\type))_{\type \in \types}$, where $\lambda(\type)\in \Delta(M^{\circ}(\type))$ for every $\type \in \types$, such that $u_\firm=\exp_{\pi}[\exp_{\lambda(\type)}[u_\firm (m, \type) ]] > \underline{u}_\firm^*$ for every $\firm$. For every firm $\firm\in \firms$, define
\[
	u^\firm = \epsilon\exp_\pi [\exp_{\lambda(\type)}[u(\zeta^\firm(m), \type)]]  + (1-\epsilon)u.
\]

First observe that because $u_\firm > \underline{u}^*_\firm \geq 0$, there exist $\type \in \types$ and $m=(\phi, {\wage}) \in M^{\circ}(\type)$ such that $\lambda(\type)[m]>0$ and $\phi(\firm) \neq \emptyset$; by the construction of $\{\zeta^\firm(m)\}_{\firm \in \firms}$, $u_\firm(\zeta^\firm(m), \type) < u_\firm(m, \type) = u_\firm(\zeta^{\firm'}(m), \type)$, so $u^\firm_\firm < u_\firm$ and $u^\firm_\firm < u^{\firm'}_\firm$ if $\epsilon> 0$. Second, every $u^\firm$ can be written as
\[u^\firm=\exp_{\pi}[\exp_{\lambda^\firm(\type)}[u_\firm (m, \type) ]],\]
where
\begin{equation}\label{equation: transformed lambda}
    \lambda^\firm(\type)=\epsilon\lambda(\type)\circ (\zeta^\firm)^{-1}+(1-\epsilon)\lambda(\type) \quad \firm\in\firms.    
\end{equation}
Note that the support of every $\lambda^\firm(\type)$ is also bounded. Finally, for small enough $\epsilon>0$, $u^\firm_\firm > \underline{u}^*_\firm$. Therefore, there must exist $\{u^\firm\}_{\firm \in \firms} \subseteq \mathcal{U}^{*}$ such that for all distinct $\firm,\firm'\in \firms$,
\[
	u^\firm_\firm < u_\firm \quad \text{and} \quad
	u^\firm_\firm < u^{\firm'}_\firm.
\]
\end{proof}

\begin{proof}[\bf Proof of \cref{prop: folk theorem}(ii)]\label{proof: folk theorem part ii}
It suffices to show that for any self-enforcing matching process $\mu$,
\[U_\firm(h\,|\,\mu) \ge (1-\delta)\underline{u}_\firm(\type) + \delta \underline{u}^*_\firm \quad \forall f\in \firms, \history\in \histories.\]

Suppose by contradiction that $U_\firm(h\,|\,\mu) < (1-\delta)\underline{u}_\firm(\type) + \delta \underline{u}^*_\firm$ for some $\firm\in \firms$ and $h \in \histories$. We will show that firm $\firm$ has a feasible profitable deviation plan from $\mu$, so $\mu$ cannot be self-enforcing.

Consider the following deviation plan $(d, \eta )$ from $\mu$: For all ex post histories that precede $\history$, the deviation plan is consistent with the matching process $\mu$. For any $\history' = (\overline{\history}', \type', \omega')\in \histories$ that succeeds $\history$ (including $\history$ itself),
\[
    d(\history') = \argmax_{A \subseteq \workers \backslash  \bigcup_{\firm'\in \firms}{\mu}(\firm' |\history') } \hchi(\firm, A , \type').
\] 
If $d(\history') \ne \emptyset$, then
\[
    \eta_{ \worker}(\history') = 
    \begin{cases}
    - \tilde{v}_\worker(\firm, d(\history')) + \frac{1}{2\abs{d(\history')} } \left[ (1-\delta)\underline{u}_\firm(\type) + \delta \underline{u}^*_\firm -U_\firm(h\,|\,\mu)\right] & \text{ if } \worker \in d(\history') ; \\
    0 & \text{ otherwise.}
    \end{cases}
\]
If $d(\history') = \emptyset$, define
\[
    \eta_{ \worker}(\history') = 0 \quad \forall \worker \in \workers.
\]
Note that by construction, $\hchi(\firm, d(\history'), \type') \geq \underline{u}_\firm(\type')$ because $\bigcup_{\firm'\in \firms}{\mu}(\firm' |h)  \leq Q$.

We first verify that deviation plan $(d, \eta )$ is feasible. At any history $\history'$ that succeeds $\history$, any $\worker\in d(\history')$ is unmatched under $\mu(\history')$ by construction, which means $v_\worker(\mu(\history'),\type')=0$. Meanwhile, according to the deviation plan,
\begin{align*}
	& \tilde{v}_\worker(\firm, d(\history')) + \eta_{ \worker}(\history') \\
        = \; & \tilde{v}_\worker(\firm, d(\history')) - \tilde{v}_\worker(\firm, d(\history')) + \frac{1}{2\abs{d(\history')} } \left[ (1-\delta)\underline{u}_\firm(\type) + \delta \underline{u}^*_\firm -U_\firm(h\,|\,\mu)\right] \\
	= \; & \frac{1}{2\abs{d(\history')} } \left[ (1-\delta)\underline{u}_\firm(\type) + \delta \underline{u}^*_\firm -U_\firm(h\,|\,\mu)\right] \\
        > \; & 0. 
\end{align*}
So at every possible history $\history'$, every worker in $d(\history')$ finds herself strictly better off by joining the deviation, which ensures the feasibility of $(d, \eta )$.
	
To see that $(d, \eta )$ is profitable, observe that at every history $\history'$ after $\history$, firm $\firm$'s stage-game payoff from the manipulated static matching $ [ \mu(\history'), ( \firm, d(\history'), \eta(\history') )] $ is
\begin{align*}
    & \tilde{u}_\firm(d(\history'), \type') - \sum_{\worker \in d(\history')} \eta_{ \worker}(\history') \\
    = \; & \tilde{u}_\firm(d(\history'), \type') + \sum_{\worker \in d(\history'))}  \tilde{v}_\worker(\firm, d(\history')) - \frac{1}{2} \left[ (1-\delta)\underline{u}_\firm(\type) + \delta \underline{u}^*_\firm -U_\firm(h\,|\,\mu)\right] \\
    \ge \; & \underline{u}_\firm(\type') - \frac{1}{2} \left[ (1-\delta)\underline{u}_\firm(\type) + \delta \underline{u}^*_\firm -U_\firm(h\,|\,\mu)\right].
\end{align*}
Since this is true for every history $\history'$ after $\history$, $\firm$'s total discounted payoff at ex post history $\history$ from the deviation plan is
\begin{align*}
    U_\firm(h\, |\, [ \mu, (\firm, d, \eta )] ) & = (1-\delta)\underline{u}_\firm(\type) + \delta \underline{u}^*_\firm- \frac{1}{2} \left[ (1-\delta)\underline{u}_\firm(\type) + \delta \underline{u}^*_\firm -U_\firm(\history\,|\,\mu)\right] \\
    & = \frac{1}{2} \left[(1-\delta)\underline{u}_\firm(\type) + \delta \underline{u}^*_\firm +U_\firm(\history\,|\,\mu)\right]\\
    & > \frac{1}{2} [(U_\firm(\history\,|\,\mu) +U_\firm(\history\,|\,\mu)]\\
    & = U_\firm(\history\,|\,\mu).
\end{align*}
Hence, deviation plan $(d, \eta )$ is both feasible and profitable for firm $\firm$, which contradicts the self-enforcement of $\mu$. Hence, 
$$U_\firm(h\,|\,\mu) \geq (1-\delta)\underline{u}_\firm(\type) + \delta \underline{u}^*_\firm\quad\forall \firm\in \firms, \history \in \histories.$$
\end{proof}

\subsection{Relaxing \texorpdfstring{\cref{assumption: generic}}{Assumption 2}}
In this section, we explain how our results continue to hold when \cref{assumption: generic} is not satisfied: There exists some firm $\firm$ such that for every $\type \in \types$ with $\pi(\type)>0$, we have $\overline{u}_\firm(\type) = \underline{u}_{\firm}(\type)$. In this case, other players' matching decisions have \emph{no impact} on the maximum static payoff firm $\firm$ can derive, since it can always turn to the unmatched workers and extract their surpluses. This means that future punishment cannot affect the matching behavior of $\firm$ via dynamic incentives, nor does $\firm$ find it beneficial to participate in any punishment scheme of other firms. Therefore, we can treat such a firm as ``inactive'' in our analysis, assign it the maximum payoff it can receive in every period, and ignore it for the rest of our analysis. An iteration is needed to identify all such firms that cannot be incentivized dynamically.

Formally, for a set of firms $\firms'\subseteq \firms$, denote by $Q(\firms')\equiv\sum_{\firm \in \firms'} q_\firm$ the total hiring capacity of $\firms'$. When all firms in $\firms \backslash \firms'$ are inactive, the effective minmax payoff of a firm $\firm \in \firms'$ at $\type$ is
\[
\underline{u}_{\firm}(\firms', \type) \equiv \min_{W' \subseteq \workers, \abs{W'} \leq Q(\firms')}\, \max_{ W \subseteq \workers\backslash W' , \abs{W}\leq q_{\firm}} \,\hchi({\firm},W, \type).
\]
Using a similar argument as in \cref{lemma: equiv minmax}, we can show that this is exactly firm $\firm$'s payoff from ``best responding'' to the worst punishment by firms in $\firms'$, while firms in $\firms \backslash \firms'$ leave no surplus to their employees. Note that for each $\type$, the value $\underline{u}_{\firm}(\firms', \type)$ weakly increases as $\firms'$ becomes smaller.

\begin{definition}
A hierarchical partition $\mathcal{P}\equiv \{\mathcal{P}_1, \ldots, \mathcal{P}_N, \mathcal{R}\}$ over firms $\firms$ is induced by the following procedure. Initialize $\mathcal{P}_0\equiv \emptyset$. For $n \geq 1$:
\begin{itemize}[leftmargin=0.5cm]
    \item If $\big\{\firm \in \firms \backslash \bigcup_{k=0}^{n-1}\mathcal{P}_k :  \overline{u}_\firm(\type) = \underline{u}_{\firm}(\firms \backslash \bigcup_{k=0}^{n-1}\mathcal{P}_k, \type) \; \forall \theta\big\} \neq \emptyset$, let this set be $\mathcal{P}_n$. Assign $n = n + 1$ and continue;
    \item If $\big\{\firm \in \firms \backslash \bigcup_{k=0}^{n-1}\mathcal{P}_k :  \overline{u}_\firm(\type) = \underline{u}_{\firm}(\firms \backslash \bigcup_{k=0}^{n-1}\mathcal{P}_k, \type) \; \forall \theta\big\} = \emptyset$, let $\mathcal{R} = \firms \backslash \bigcup_{k=0}^{n-1}\mathcal{P}_k$ and stop.
\end{itemize}
\end{definition}

Intuitively, each $\mathcal{P}_n$ consists of firms that cannot be punished in the matching process without cooperation from those in $\bigcup_{k=0}^{n-1}\mathcal{P}_k$. If $\mathcal{R} \neq \emptyset$, by construction, $\overline{u}_\firm(\type)>\underline{u}_\firm(\mathcal{R}, \type)$ for every $\firm \in \mathcal{R}$ and $\type \in \types$. Let 
\[
\overline{u}^*_\firm \equiv \exp_{\pi}[\overline{u}_\firm(\type)] \quad \forall \firm \in \firms \backslash \mathcal{R},
\]
and
\[
\underline{u}^*_\firm(\mathcal{R})\equiv \exp_{\pi}[\underline{u}_{\firm}(\mathcal{R},\type)] \quad \forall \firm \in \mathcal{R}.
\]

A generalized version of \cref{prop: folk theorem} can be stated as follows.

\begin{proposition} \label{prop: folk theorem general} 
(i) If $u\in \mathcal{U}^*$ satisfies $u_{\firm} = \overline{u}^*_\firm$ for all $\firm \in \firms \backslash \mathcal{R}$ and $u_{\firm} > \underline{u}^*_\firm(\mathcal{R})$ for all $\firm\in \mathcal{R}$, there is a $\underline{\delta}\in (0,1)$ such that for every $\delta\in(\underline{\delta},1)$, there exists a self-enforcing matching process with firms' discounted payoffs $u$.
(ii) Suppose $\mu$ is a self-enforcing matching process for a given $\delta\in (0,1)$. For every ex ante history $\overline{\history}\in \fhistories$, firms' discounted payoff profile satisfies $\big( U_{\firm}(\overline{\history}\,|\,\mu) \big)_{\firm\in \firms}\in \mathcal{U}^*$, $U_f(\overline{\history}\,|\,\mu) = \overline{u}^*_\firm$ for every $\firm\in \firms \backslash \mathcal{R}$, and  $U_f(\overline{\history}\,|\,\mu) \geq \underline{u}^*_\firm(\mathcal{R})$ for every $\firm\in \mathcal{R}$.
\end{proposition}

\begin{proof}
For part (i), let $(\lambda(\type))_{\type\in \types}$ be the tuple of lotteries such that $\lambda(\type)\in \Delta(M^{\circ}(\type)) \;\forall\type\in \types$, and $u = \exp_{\pi}[\exp_{\lambda(\type)}[u (m, \type) ]]$. There are three cases to consider.

\underline{Case 1:} $\mathcal{R}=\emptyset$. If a firm $\firm$ receives $u_\firm = \overline{u}^*_\firm$ on average, it is necessary that this firm receives the highest possible payoff $\overline{u}_\firm(\type)$ at every realization of $\type$. This in turn implies that, for each $\type$, $\lambda(\type)$ only assigns positive probability to stage matchings that are \emph{stable} in a static sense. Therefore, a matching process that recommends according to $(\lambda(\type))_{\type\in \types}$ in every period is self-enforcing.

\underline{Case 2:} $\abs{\mathcal{R}}=1$. Let $\firm$ denote the single firm in $\mathcal{R}$, and let $(\underline{m}_\firm(\type))_{\type \in \types}$ be the matchings that give $\firm$ the most severe punishment by $\firm$ itself, while all other firms receive the highest possible payoff $\overline{u}_\firm(\type)$. Consider the following matching process:
\begin{itemize}
\item[(I)] Match according to $\lambda(\cdot)$ if $\lambda(\cdot)$ was followed in the last period or $\underline{m}_\firm(\cdot)$ was followed for $L$ periods;
\item[(II)] If firm $\firm$ deviates from (I), match according to $\underline{m}_\firm(\cdot)$
for $L$ periods.
\end{itemize}
If firm $\firm$ deviates from (II), restart (II).

It is easy to check that when $L$ is sufficiently large, firm $\firm$ has no incentive to deviate in either phase, since $\delta\rightarrow 1$. All other firms have no incentive to deviate, since they already receive maximum stage payoff in every period.

\underline{Case 3:} $\abs{\mathcal{R}}\geq 2$. The proof for this case essentially follows the one for \cref{prop: folk theorem} with proper adjustments.

For part (ii), by construction, any firm $\firm \in \mathcal{P}_1$ can secure a stage payoff $\overline{u}_\firm(\type)$ by deviation at every $\type$ regardless of the stage matching. Taking expectation yields $U_f(\overline{\history}\,|\,\mu) = \overline{u}^*_\firm$ for every $\firm \in \mathcal{P}_1$.

Suppose $U_{f'}(\overline{\history}\,|\,\mu) = \overline{u}^*_{\firm'}$ for every $\firm' \in \bigcup_{k=1}^{n}\mathcal{P}_k$ with $n< N$. Then all firms in $\bigcup_{k=1}^{n}\mathcal{P}_k$ offer zero wages to their employees. By construction, each firm $\firm \in \mathcal{P}_{n+1}$ can secure a stage payoff $\overline{u}_\firm(\type)$ by deviation at every $\type$ regardless of how firms in $\firms \backslash \bigcup_{k=1}^{n}\mathcal{P}_k$ are matched in each period. Taking expectation yields $U_f(\overline{\history}\,|\,\mu) = \overline{u}^*_\firm$ for every $\firm \in \mathcal{P}_{n+1}$. By induction, the equality holds for all $\firm\in \firms \backslash \mathcal{R}$.

By definition of the effective minmax payoff, every firm $\firm\in \mathcal{R}$ can secure $\underline{u}_f(\mathcal{R}, \type)$ in each period by deviating with workers. Taking expectation over $\type$ yields $U_f(\overline{\history}\,|\,\mu) \geq \underline{u}^*_\firm(\mathcal{R})$ for these firms. Rigorous proof can be adapted from that of \cref{prop: folk theorem}(ii) in \cref{proof: folk theorem part ii}.
\end{proof}

We next define a random serial dictatorship with respect to the hierarchical partition $\mathcal{P}=\{\mathcal{P}_1, \ldots, \mathcal{P}_N, \mathcal{R}\}$. To do so, we first introduce a subset of orderings $\Gamma_{\mathcal{P}}\subseteq \Gamma$ that contains those that (i) give firms in $\mathcal{R}$ the highest priorities and (ii) rank firms in $ \mathcal{P}_n$ higher than those in $\mathcal{P}_k$ if $n> k$. That is,
\[
    \Gamma_{\mathcal{P}} \equiv \left\{ \gamma \in \Gamma : \begin{array}{c}
    \gamma(\mathcal{R})=\{1,2,\ldots, \abs{\mathcal{R}}\}, \text{ and}\\
    \text{if $\firm \in \mathcal{P}_k$, $\firm' \in \mathcal{P}_n$, and $k<n$, then $\gamma(\firm)>\gamma(\firm')$}
    \end{array}\right\}.
\]

The stage matching $\hat{m}(\type, \gamma)$ induced by a serial dictatorship according to $\gamma$ is defined as in \cref{definition: serial dictatorship}. By using a random serial dictatorship restricted to $\Gamma_{\mathcal{P}}$, the following lemma generalizes \cref{lemma: nonemptiness} and shows that \cref{prop: folk theorem general}(i) is not vacuously true.

\begin{lemma} \label{lemma: nonemptiness general}
For every firm $\firm \in  \mathcal{R}$,
\[
\frac{1}{\abs{\Gamma_{\mathcal{P}}}}\sum_{\gamma\in \Gamma_{\mathcal{P}}} \exp_{\pi}\Big[{u}_{\firm}\big(\,\hat{m}(\type,\gamma),\type\,\big)\Big] > \underline{u}_{\firm}^*(\mathcal{R}).
\]
For every firm $\firm \in \firms \backslash \mathcal{R}$,
\[
\frac{1}{\abs{\Gamma_{\mathcal{P}}}}\sum_{\gamma\in \Gamma_{\mathcal{P}}} \exp_{\pi}\Big[{u}_{\firm}\big(\,\hat{m}(\type,\gamma),\type\,\big)\Big] = \overline{u}^*_\firm.
\]

\begin{proof}
    The proof of the first statement follows that of \cref{lemma: nonemptiness}.
    
    For the second statement, take any $\firm \in \mathcal{P}_{n}$, $n=1,2,\ldots, N$. By definition, for every $\gamma \in \Gamma_{\mathcal{P}}$, we have $\abs{{W}^{\#}_{\gamma(\firm)}}<  Q(\firms\backslash \bigcup_{k=0}^{n-1}\mathcal{P}_k)$. This means
    \begin{align*}
        {u}_{\firm}\big(\,\hat{m}(\type,\gamma),\type\,\big)&= \max_{W\subseteq \workers\backslash{W}^{\#}_{\gamma(\firm)}, |W|\le q_{\firm}} \hchi(\firm,W,\type)\\
        &\geq \min_{W' \subseteq \workers, \abs{W'} \leq Q(\firms\backslash \bigcup_{k=0}^{n-1}\mathcal{P}_k)}\, \max_{ W \subseteq \workers\backslash W' , \abs{W}\leq q_{\firm}} \,\hchi({\firm},W, \type)\\
        &=\overline{u}_\firm(\type), \quad \forall \type\in\types,
    \end{align*}
    where the last equality comes from the definition of $\mathcal{P}_{n}$. Taking expectation over $\types$ gives $\exp_{\pi}\big[{u}_{\firm}\big(\,\hat{m}(\type,\gamma),\type\,\big)\big] = \overline{u}^*_\firm$, which suffices for the second statement to hold.
\end{proof}
	
\end{lemma} 

In view of \cref{lemma: nonemptiness general}, we have established \cref{prop: existence} without \cref{assumption: generic}. That is, a self-enforcing matching process exists when firms are sufficiently patient.

\end{document}